\documentclass[12pt,a4paper,reqno]{amsart}
\usepackage{amsmath}
\usepackage{amsfonts}
\usepackage{amssymb}
\usepackage{amsthm}
\usepackage{enumerate}
\usepackage{centernot}
\usepackage{fullpage}
\usepackage{mathrsfs}
\usepackage{graphicx}
\usepackage{color}

\newcommand{\upperRomannumeral}[1]{\uppercase\expandafter{\romannumeral#1}}

\theoremstyle{plain}
  
  \newtheorem{claim}[]{Claim}
  \newtheorem{proposition}[]{Proposition}
  \newtheorem{lemma}[]{Lemma}
  \newtheorem{theorem}[]{Theorem}
  \newtheorem{corollary}[]{Corollary}
  
  \newtheorem{remark}[]{Remark}

\title[Lee-Yang zeros]{Motion of Lee-Yang zeros}

\author{Qi Hou}
\address{Yanqi Lake Beijing Institute of Mathematical Sciences and Applications, No.11 Yanqi Lake West Road, Beijing 101408, China.}
\email{houqi@bimsa.cn}
\author{Jianping Jiang}
\address{Yau Mathematical Sciences Center, Tsinghua University, Beijing 100084, China.}
\email{jianpingjiang@tsinghua.edu.cn}
\author{Charles M. Newman}
\address{Courant Institute of Mathematical Sciences, New York University,
	251 Mercer st, New York, NY 10012, USA, \& NYU-ECNU Institute of Mathematical
	Sciences at NYU Shanghai, 3663 Zhongshan Road North, Shanghai 200062, China.}
\email{newman@cims.nyu.edu}

\begin{document}
\begin{abstract}
We consider the zeros of the partition function of the Ising model 
with ferromagnetic pair interactions and complex external field. Under the assumption 
that the graph with strictly positive interactions is connected, we vary  the interaction 
(denoted by $t$) at a fixed edge. It is already known that each zero is monotonic 
(either increasing or decreasing) in $t$; we prove that its motion is local: the entire 
trajectories of any two distinct zeros are disjoint. If the underlying graph is a complete graph and all interactions 
take the same value $t\geq 0$ (i.e., the Curie-Weiss model), we prove that all the 
principal zeros (those in $i[0,\pi/2)$) decrease strictly in $t$.
\end{abstract}

\maketitle

\section{Introduction and main results}
\subsection{Overview}
The Ising model is one of the most studied models in statistical physics. In 1952, 
Yang and Lee \cite{YL52,LY52} studied the partition function of the Ising model. One surprising result coming out of their study was that all zeros of the partition function lie on the unit circle in the complex fugacity plane (the imaginary axis in the complex external field plane). This result has been extended to other systems \cite{Rue71,SG73,New74,LS81,BBCKK04,BBCK04}. The Lee-Yang circle theorem has been widely applied to study properties of phase transitions in statistical physics. For example, it was applied to prove the mass gap \cite{PL74,GRS75}, correlation inequalities \cite{New75}, inequalities for critical exponents \cite{Sok81}, and scaling limits of total magnetization \cite{CJN21}; see also \cite{PZWCDL15} for an interesting experimental study. 
We refer to \cite{BDL05} and the many references therein for a review of the Lee-Yang theory. 

The original goal of the Lee-Yang program was to understand phase transitions of 
models in statistical physics by studying directly the zeros of their partition 
functions. That is, the complex singularities of the free energy (which 
for a finite system are exactly the zeros of the partition function) may approach 
the real axis in the thermodynamic limit. It is believed that in the infinite volume 
limit, the distribution of Lee-Yang zeros has a density $g_T(\theta)$ where $T$ 
is the temperature and $\theta\in[-\pi,\pi]$ corresponds to a point on the unit circle. 
The existence of $g_T$ was rigorously proved only in \cite{BBCKK04} for very low temperature 
$T\ll T_c$ where $T_c$ is the critical temperature. It was proved in \cite{CJN21} 
that for each $T>T_c$, the limit distribution of Lee-Yang zeros has no support 
in a neighborhood of $\theta=0$, and thus $g_T(\theta)=0$ for all small 
$|\theta|$. Suppose that the support of $g_T$ is $\{\theta:\theta_0(T)\leq |\theta|\leq \pi\}$ 
where $\theta_0(T)\in(0,\pi)$ for $T>T_c$. Then the Yang-Lee edge singularity  
\cite{KG71,Fis78,Car85} describes the expected power-law behavior of $g_T$ 
near the critical value $\theta_0(T)$, i.e., $g_T(\theta)\sim|\theta-\theta_0(T)|^{\sigma}$ 
for some critical exponent $\sigma$ (which depends on the spatial dimension); at $T=T_c$, 
$\theta_0(T_c)$ is expected to be $0$ and $\sigma$ is expected to equal the critical 
exponent $\delta$ which governs the behavior of magnetization as a function of the 
external field; in the physics literature, there is also a predicted scaling behavior 
for $g_T(\theta)$ when $T$ is close to $T_c$ which is related to the critical 
exponents $\beta$ and $\gamma$ , see, e.g., (53) of \cite{BDL05}. If all those 
critical exponents related to $g_T$ can be determined (eventually rigorously), then 
one would have a full understanding of the Ising model using scaling relations to obtain other critical exponents.

Despite the elegant picture described in the previous paragraph, rigorous results 
about the distribution of the Lee-Yang zeros are very rare in both the discrete 
(i.e., on finite graphs) and continuum (i.e., in the thermodynamic limit) settings. 
Two very nice results about the limiting distribution of Lee-Yang zeros are \cite{CHJR19} 
for the Cayley tree and \cite{Kab22} for the complete graph. But there is almost no 
rigorous result concerning the Lee-Yang zeros on $\mathbb{Z}^d$ if the temperature 
is near the critical value, which is expected to be the most interesting case. So it 
seems fair to say that the original Lee-Yang program is far from complete. 
In this paper, we consider the Lee-Yang zeros of the Ising model with ferromagnetic 
pair interactions defined on a finite graph. Under the assumption that the graph 
with strictly positive interactions is connected, we study the motion of those 
zeros as the interactions vary. If one varies only a single interaction, $t$, at a 
fixed edge, it has been proved in \cite{NG83} that each zero moves monotonically, 
i.e., each zero can behave in only one of the three ways: constant, strictly increasing, or 
strictly decreasing; we prove (see Theorem \ref{thm:t}) that for any two distinct zeros 
at $t=0$ (say $x_k$ and $x_j$ with $x_k(0)\neq x_j(0)$), their entire trajectories 
are disjoint: $x_k([0,\infty))\cap x_j([0,\infty))=\emptyset$. We remark that it has 
been proved in~\cite{CJN22} that the first zero is always decreasing in $t\geq0$ 
regardless of the connectedness of the underlying graph; as pointed out in \cite{NG83}, 
it is possible that the second zero is increasing in $t\geq 0$ (see Remark \ref{rem:inc}). 

In \cite{NG83}, Nishimori and Griffiths conjectured that all principal zeros decrease 
in $t>0$ if the pair interaction on each edge is $t$ and the graph is connected. 
Our second main result (see Theorem \ref{thm:K_n}) supports this conjecture by considering the Ising 
model on the complete graph of $n$ vertices (i.e., the Curie-Weiss model) and proving 
that all the principal zeros (those that lie in $i[0,\pi/2)$) decrease strictly in 
$t\geq 0$. We should mention that the limit distribution of the Lee-Yang zeros 
for the Curie-Weiss model has been identified recently in \cite{Kab22}. We hope 
that the results in the current paper may lead to further development
of the Lee-Yang program.

\subsection{Main results}
Let $G=(V,E)$ be a finite graph with $V$ the set of vertices and $E$ the set of edges. 
The Ising model on $G$ with \textit{ferromagnetic pair interactions} (or \textit{couplings}) 
$\mathbf{J}:=(J_{uv})_{uv\in E}$ where $J_{uv}\in [0,\infty)$ for each $uv\in E$ 
and external field $h\in {\mathbb R}$ is defined by the probability measure 
$\mathbb{P}_{G,\mathbf{J},h}$ on $\{-1,+1\}^{V}$ such that
\begin{equation}\label{eq:Isingdef}
	\mathbb{P}_{G,\mathbf{J},h}(\sigma)=\frac{\exp\left[\sum_{uv\in E}J_{uv}\sigma_u\sigma_v+h\sum_{u\in V}\sigma_u\right]}{Z_{G,\mathbf{J},h}}, \sigma\in\{-1,+1\}^{V},
\end{equation}
where $Z_{G,\mathbf{J},h}$ is the partition function that makes 
$\mathbb{P}_{G,\mathbf{J},h}(\sigma)$
a probability measure. We are interested in the zeros of 
$Z_{G,\mathbf{J},h}$ as a function of $h\in\mathbb{C}$. Note that
\begin{align}
	Z_{G,\mathbf{J},h}:=&\sum_{\sigma\in\{-1,+1\}^V}\exp\left[\sum_{uv\in E}J_{uv}\sigma_u\sigma_v+h\sum_{u\in V}\sigma_u\right]\\
	=&\exp\left[h|V|\right]\sum_{\sigma\in\{-1,+1\}^V}\exp\left[\sum_{uv\in E}J_{uv}\sigma_u\sigma_v+h\sum_{u\in V}(\sigma_u-1)\right].
\end{align} 
Let $z=e^{-2h}$. Then it is clear that $Z_{G,\mathbf{J},h}$ divided by $\exp[h|V|]$ 
is a polynomial in $z$ with degree $|V|$. So by the fundamental theorem of algebra, 
$Z_{G,\mathbf{J},h}$ has exactly $|V|$ complex roots (in the variable $z$). The 
Lee-Yang circle theorem \cite{LY52} says that these $|V|$ roots are all on the unit 
circle. So we may assume that these roots are
\begin{equation}\label{eq:thetadef}
	\exp(i\theta_1), \exp(i\theta_2), \dots, \exp(i\theta_{|V|}) \text{ with }0<\theta_1\leq\theta_2\leq\dots\leq\theta_{|V|}<2\pi,
\end{equation}
where for the strict inequalities, we have used the fact that $Z_{G,\mathbf{J},0}>0$. By spin-flip symmetry, 
those $|V|$ roots are symmetric with respect to the real axis. 

Therefore, 
if $n:=|V|$ is even, all zeros of $Z_{G,\mathbf{J},h}$ as a function of $h$ for fixed $\mathbf{J}$ are
\begin{align}\label{eq:evenzeros}
	&i\left(\pm\frac{\theta_1(\mathbf{J})}{2}+m\pi\right), i\left(\pm\frac{\theta_2(\mathbf{J})}{2}+m\pi\right),\dots,i\left(\pm\frac{\theta_{n/2}(\mathbf{J})}{2}+m\pi\right)\nonumber\\
	&\text{ with }0<\theta_1(\mathbf{J})\leq\theta_2(\mathbf{J})\leq\dots\leq\theta_{n/2}(\mathbf{J})\leq\pi, m\in\mathbb{Z};
\end{align}
if $n$ is odd, all zeros of $Z_{G,\mathbf{J},h}$ as a function of $h$ for fixed $\mathbf{J}$ are
\begin{align}\label{eq:oddzeros}
	&i\left(\pm\frac{\theta_1(\mathbf{J})}{2}+m\pi\right), \dots,i\left(\pm\frac{\theta_{(n-1)/2}(\mathbf{J})}{2}+m\pi\right),i\left(\frac{\theta_{(n+1)/2}(\mathbf{J})}{2}+m\pi\right)\nonumber\\
	&\text{ with }0<\theta_1(\mathbf{J})\leq\theta_2(\mathbf{J})\leq\dots\leq\theta_{(n-1)/2}(\mathbf{J})\leq\pi=\theta_{(n+1)/2}(\mathbf{J}), m\in\mathbb{Z}.
\end{align}
We are interested in the motion of zeros when one increases the interaction at a 
fixed edge $u_0v_0\in E$. More precisely,  we define
\begin{equation}\label{eq:Z_tdef}
	Z_{G,t}(x):=\sum_{\sigma\in\{-1,+1\}^{V}}\exp\left[t\sigma_{u_0}\sigma_{v_0}+\sum_{uv\in E}J_{uv}\sigma_u\sigma_v+ix\sum_{u\in V}\sigma_u\right], t\geq 0,
\end{equation}
where we have dropped the dependence 
on $\mathbf{J}$ from the notation for the
partition function. This should 
cause no confusion since we fix $\mathbf{J}$ in Theorem \ref{thm:t} below. 
We want to study the motion of zeros of $Z_{G,t}(x)$ as a function of $t$. Note that 
we have written the external field $h$ as $ix$ so that all zeros of $Z_{G,t}(x)$ 
are real and they are
\begin{align}
	&\pm x_1(t)+m\pi, \pm x_2(t)+m\pi, \dots, \pm x_{n/2}(t)+m\pi \text{ with } m\in\mathbb{Z} \text{ and }\nonumber\\
	&\qquad 0<x_1(t)\leq x_2(t)\leq\dots\leq x_{n/2}(t)\leq\pi/2  \text{ if }n \text{ is even}, \label{eq:evenzeros1}\\
	&\pm x_1(t)+m\pi, \pm x_2(t)+m\pi, \dots, \pm x_{(n-1)/2}(t)+m\pi, \pi/2+m\pi \text{ with }m\in\mathbb{Z}\text{ and }\nonumber\\
	& \qquad0<x_1(t)\leq x_2(t)\leq\dots\leq x_{(n-1)/2}(t)\leq\pi/2  \text{ if } n\text{ is odd}\label{eq:oddzeros1} .
\end{align}
In terms of the definitions in \eqref{eq:evenzeros} and \eqref{eq:oddzeros},
\begin{equation}\label{eq:xandtheta}
	x_k(t)=\theta_k(\tilde{\mathbf{J}})/2, 
	\forall k=1,2,\dots,n/2\text{ if }n\text{ is even }, 
	\forall k=1,2,\dots,(n+1)/2\text{ if }n\text{ is odd },
\end{equation}
where $\tilde{\mathbf{J}}:=(J_{uv})_{uv\in E}$,  $\tilde{J}_{uv}:=J_{uv}$ if $uv\neq u_0v_0$ and $\tilde{J}_{u_0v_0}:=J_{u_0v_0}+t$. For a fixed $\mathbf{J}$, its associated subgraph $G_{>0}:=(V,E_{>0})$ of $G=(V,E)$ is defined by
\begin{equation}\label{eq:E+}
	E_{>0}:=\{uv: J_{uv}>0, uv\in E \}.
\end{equation}
Our first main result is
\begin{theorem}\label{thm:t}
Let $G=(V,E)$ be a finite graph and $\mathbf{J}$ be ferromagnetic pair interactions on $G$. Suppose that  $G_{>0}=(V,E_{>0})$ defined by \eqref{eq:E+} is connected. If $n=|V|$ is even, then the first $n/2$ positive zeros of $Z_{G,t}(x)$ satisfy
\begin{equation}\label{eq:xsimple}
	0<x_1(t)<x_2(t)<\dots<x_{n/2}(t)<\frac{\pi}{2}, t\in[0,\infty).
\end{equation}
Moreover, for each $k\in\{1,2,\dots,n/2\}$, exactly one of the following three cases occurs:
\begin{enumerate}[a)]
	\item $x_k(t)\equiv x_k(0)$ for each $t\in[0,\infty)$,
	\item $x_k(t)$ is strictly decreasing in $t\in[0,\infty)$ and $x_k(t)\in\left(x_{k-1}(0),x_k(0)\right)$ for each $t\in(0,\infty)$,
	\item $x_k(t)$ is strictly increasing in $t\in[0,\infty)$ and $x_k(t)\in\left(x_{k}(0),x_{k+1}(0)\right)$ for each $t\in(0,\infty)$,
\end{enumerate}
where we set $x_0(0):=0$ and $x_{n/2+1}(0):=\pi/2$. 

If $n=|V|$ is odd, then the first $(n+1)/2$ positive zeros of $Z_{G,t}(x)$ satisfy
\begin{equation}\label{eq:xsimple1}
	0<x_1(t)<x_2(t)<\dots<x_{(n-1)/2}(t)<\frac{\pi}{2}=x_{(n+1)/2}(t), t\in[0,\infty).
\end{equation}
Moreover, for each $k\in\{1,2,\dots,(n-1)/2\}$, exactly one of the following three cases occurs:
\begin{enumerate}[a)]
	\item $x_k(t)\equiv x_k(0)$ for each $t\in[0,\infty)$,
	\item $x_k(t)$ is strictly decreasing in $t\in[0,\infty)$ and $x_k(t)\in\left(x_{k-1}(0),x_k(0)\right)$ for each $t\in(0,\infty)$,
	\item $x_k(t)$ is strictly increasing in $t\in[0,\infty)$ and $x_k(t)\in\left(x_{k}(0),x_{k+1}(0)\right)$ for each $t\in(0,\infty)$,
\end{enumerate}
where we set $x_0(0):=0$.
\end{theorem}
\begin{remark}
	Since all zeros of $Z_{G,t}$ (see \eqref{eq:evenzeros1} and \eqref{eq:oddzeros1}) 
are reflections and periodic translations of those zeros considered in the theorem, 
Theorem~\ref{thm:t} implies that each zero of $Z_{G,t}$ behaves in one of three ways
as $t$ increases: is constant; decreases strictly; increases strictly. This has 
already been proved in Theorem 4.8 of \cite{NG83}. Combining with 
Corollary~\ref{cor:simple} below, Theorem~\ref{thm:t}  yields the following new result: 
for any two zeros $x_k$ and $x_j$ satisfying $x_k(0)\neq x_j(0)$, their entire 
trajectories are disjoint, i.e., $x_k([0,\infty))\cap x_j([0,\infty))=\emptyset$.
\end{remark}
\begin{remark}\label{rem:inc}
	It was proved in \cite{CJN22} that $x_1(t)$ is always decreasing in $t\geq 0$ 
even without the assumption that $G_{>0}$ is connected. The following example shows 
that $x_2(t)$ can be increasing in $t$. Let $G:=G_4$ be the left graph in Figure \ref{fig:G_4G_5} with the interactions as indicated.
\begin{figure}
	\begin{center}
		\includegraphics{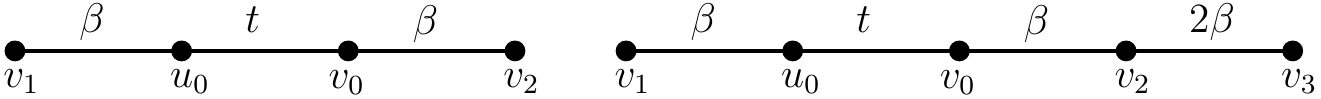}
		\caption{$G_4$ left, $G_5$ right; the interaction at each edge is indicated right above the edge.}\label{fig:G_4G_5}
	\end{center}
\end{figure}
Then a direct computation gives
	\begin{equation}
		\lim_{\beta\rightarrow\infty}x_1(t=0,\beta)=\lim_{\beta\rightarrow\infty}x_2(t=0,\beta)=\frac{\pi}{4},\lim_{\substack{t\rightarrow\infty\\\beta\rightarrow\infty}}x_1(t,\beta)=\frac{\pi}{8}, \lim_{\substack{t\rightarrow\infty\\\beta\rightarrow\infty}}x_2(t,\beta)=\frac{3\pi}{8}.
	\end{equation}
So by the continuity of $x_2$ in $t$ and $\beta$ (which follows from Hurwitz's theorem, see the proof of Lemma \ref{lem:cont} below), we may pick $\beta_0>0$ such that
\begin{equation}
	x_1(t=0,\beta=\beta_0)=x_2(t=0,\beta=\beta_0)<\frac{3\pi}{8}, \lim_{t\rightarrow\infty}x_2(t,\beta=\beta_0)>x_2(t=0,\beta=\beta_0).
\end{equation}
Then Theorem \ref{thm:t} implies that $x_2(t,\beta=\beta_0)$ is strictly increasing in $t>0$.
\end{remark}

\begin{remark}
	It is tempting to conjecture based on Theorem \ref{thm:t} 
that under its assumptions, all zeros except those trivial ones at $m\pi+\pi/2$ 
(when $n$ is odd) are strictly monotonic. The following example shows that this 
conjecture is false. Let $G:=G_5$ be the right graph in Figure~\ref{fig:G_4G_5} with the interactions as indicated.
Then the corresponding partition function is
\begin{align}
	Z_{G,t}(x)=&2e^t\left[e^{4\beta}\cos(5x)+(e^{2\beta}+e^{-2\beta}+1)\cos(3x)+(e^{2\beta}+e^{-2\beta}+e^{-4\beta}+1)\cos(x)\right]+\nonumber\\
	&2e^{-t}\left[2e^{2\beta}\cos(3x)+(e^{4\beta}+2e^{-2\beta}+e^{-4\beta}+2)\cos(x)\right].
\end{align}
One can check that if $\beta=\ln(1+\sqrt{2})/2$, then $Z_{G,t}(\arcsin(2^{-1/4}))=0$ for each $t\geq 0$. 
In particular, this implies that $\arcsin(2^{-1/4})$ is a zero of $Z_{G,t}(x)$, 
which does not change with~$t$.
\end{remark}

We next consider the Ising model on the complete graph. More precisely, we consider the partition function
\begin{equation}\label{eq:Z_ntdef}
	Z_{n,t}(x):=\sum_{\sigma\in\{-1,+1\}^{n}}\exp\left[t\left(\sum_{j=1}^n\sigma_j\right)^2+ix\sum_{j=1}^n\sigma_j\right].
\end{equation}
This partition function corresponds to the Ising model defined on the complete graph 
with $n$ vertices and with interaction $t$ at each edge; so it is the partition function for the Curie-Weiss model. 
Our second main result is
\begin{theorem}\label{thm:K_n}
	Let the zeros of $Z_{n,t}$ as a function of $x$ be defined as in \eqref{eq:evenzeros1} and \eqref{eq:oddzeros1} according to the parity of $n$. If $n$ is positive and even, then for each $k\in\{1,2,\dots,n/2\}$,
	\[x_k(t)\text{ is strictly decreasing in }t\geq 0 \text{ with }x_k(0)=\frac{\pi}{2}, \lim_{t\rightarrow\infty}x_k(t)=\frac{(2k-1)\pi}{2n};\]
	if $n$ is positive and odd, then $x_{(n+1)/2}(t)\equiv \pi/2$ and for each $k\in\{1,2,\dots,(n-1)/2\}$,
	\[x_k(t)\text{ is strictly decreasing in }t\geq 0 \text{ with }x_k(0)=\frac{\pi}{2}, \lim_{t\rightarrow\infty}x_k(t)=\frac{(2k-1)\pi}{2n}.\]
\end{theorem}

\begin{remark}
	According to Conjecture 2 in Section IVC of \cite{NG83}, Theorem \ref{thm:K_n} 
is expected to hold for any connected graph $G=(V,E)$ with $|V|=n$ and $J_e=t$ for each $e\in E$. 
For the Ising model on such a general graph, Corollary 1 in \cite{CJN22} 
implies that its first zero is strictly decreasing.
\end{remark}

The proofs for both Theorems \ref{thm:t} and \ref{thm:K_n} rely on a system of ordinary differential equations (ODEs) which governs the motion of all principal zeros. The ODEs are derived from the factorization formula for the partition function, and an extra relation for the partition function at different $t$ for the  case of Theorem \ref{thm:t}. In Section \ref{sec:t}, we prove Theorem 1 by using these ODEs and a relation among different zeros (see Lemma \ref{lem:ydifft} below). In Section~\ref{sec:K_n}, we prove that for the Curie-Weiss model, it is a monotone dynamical system (more precisely, a cooperative system), which enables us to complete the proof of Theorem~\ref{thm:K_n}.

\section{Motion with respect to one interaction}\label{sec:t}
In this section, we prove Theorem \ref{thm:t}. To simplify the notation, we write $Z_t$ for $Z_{G,t}$, which is defined in \eqref{eq:Z_tdef}. Throughout this section, $n:=|V|$ is the total number of vertices. We first prove that $Z_t(x)$ has the following factorization formula.
\begin{lemma}\label{lem:fac}
	Let $G=(V,E)$ be a finite graph and $\mathbf{J}$ be ferromagnetic pair interactions on $G$. For each $t\geq 0$, any $ x\in\mathbb{C}$,
	\begin{equation}\label{eq:fac}
		Z_t(x)=Z_t(0)\prod_{j=1}^{\infty}\left(1-\frac{x^2}{x_j^2(t)}\right)=
		\begin{cases}
			Z_t(0)\prod_{j=1}^{n/2}\frac{\sin^2x_j(t)-\sin^2x}{\sin^2x_j(t)}, &\text{if } n \text{ even}\\
			Z_t(0)\cos(x)\prod_{j=1}^{(n-1)/2}\frac{\sin^2x_j(t)-\sin^2x}{\sin^2x_j(t)}, &\text{if } n \text{ odd},
		\end{cases}
	\end{equation}
where $0<x_1(t)\leq x_2(t)\leq\dots$ are all positive zeros of $Z_t(x)$ (listed according to their multiplicities), and $\sum_{j=1}^{\infty}x_j^{-2}(t)<\infty$.
\end{lemma}
\begin{proof}
	The first equality in \eqref{eq:fac} and $\sum_{j=1}^{\infty}x_j^{-2}(t)<\infty$ follow from the Hadamard factorization theorem (see, e.g., pp. 206-212 of \cite{Ahl78}) and the fact that $Z_t(x)$ for fixed $t\geq0$ is an even entire function of order 1; see also Proposition 2 of \cite{New75}. By the product representation of sine (see, e.g., (24) on p. 197 of \cite{Ahl78}), we have
	\begin{equation}
		\sin(\pi z)=\pi z\prod_{m=1}^{\infty}\left(1-\frac{z^2}{m^2}\right), \forall z\in\mathbb{C}.
	\end{equation}
It then follows that for any $y\in\mathbb{C}\setminus\{k\pi: k\in\mathbb{Z}\}$ and any $ x\in\mathbb{C}$,
\begin{equation}\label{eq:sinprod}
	\frac{\sin(y+x)\sin(y-x)}{\sin^2 y}=\prod_{m=0}^{\infty}\left(1-\frac{x^2}{(y+\pi m)^2}\right)\prod_{m=1}^{\infty}\left(1-\frac{x^2}{(-y+\pi m)^2}\right).
\end{equation}
If $n$ is even, the first equality in \eqref{eq:fac} and \eqref{eq:evenzeros1} imply that
\begin{align}
	Z_t(x)&=Z_t(0)\prod_{j=1}^{n/2}\prod_{m=0}^{\infty}\left(1-\frac{x^2}{(x_j+\pi m)^2}\right)\prod_{m=1}^{\infty}\left(1-\frac{x^2}{(-x_j+\pi m)^2}\right)\\
	&=Z_t(0)\prod_{j=1}^{n/2}\frac{\sin(x_j+x)\sin(x_j-x)}{\sin^2x_j}=Z_t(0)\prod_{j=1}^{n/2}\frac{\cos(2x)-\cos(2x_j)}{2\sin^2x_j}\\
	&=Z_t(0)\prod_{j=1}^{n/2}\frac{\sin^2x_j(t)-\sin^2x}{\sin^2x_j(t)},
\end{align}
where we have used \eqref{eq:sinprod} in the second equality. The proof for odd $n$ is similar and uses the following product representation for cosine
\begin{equation}
	\cos(z)=\prod_{m=0}^{\infty}\left(1-\frac{z^2}{(m\pi+\pi/2)^2}\right), \forall z\in\mathbb{C}.
\end{equation}
This completes the proof of the lemma.
\end{proof}

An easy but useful property about $x_k(t)$ is that they are continuous in $t$.
\begin{lemma}\label{lem:cont}
	For any finite graph $G=(V,E)$ and any ferromagnetic pair interactions $\mathbf{J}$ on $G$, each positive zero defined in Lemma \ref{lem:fac}, $x_k(t)$ with $k\in\mathbb{N}$, is continuous in $t\in[0,\infty)$.
\end{lemma}
\begin{proof}
	This follows from Hurwitz's theorem (see, e.g., p. 4 of \cite{Mar66} or Lemma 6 of \cite{CJN21}).
\end{proof}

The next result by Nishimori and Griffiths \cite{NG83} says that all zeros of $Z_t(x)$ are simple as long as the graph associated with $\mathbf{J}$, $G_{>0}=(V,E_{>0})$ with $E_{>0}$ defined by \eqref{eq:E+}, is connected.
\begin{proposition}[\cite{NG83}]\label{prop:simple}
	Let $G=(V,E)$ be a finite graph and $\mathbf{J}$ be ferromagnetic pair interactions on $G$. Suppose that $G_{>0}=(V,E_{>0})$ defined by \eqref{eq:E+} is connected. Then for any fixed $t\in[0,\infty)$, all zeros of $Z_t(x)$ are simple.
\end{proposition}
\begin{proof}
	This follows from Lemma 4.2 and Theorem 3.6(i) (with $g_j=1$ for each $j$ and $A=\emptyset$) of \cite{NG83}.
\end{proof}

\begin{remark}
	Proposition \ref{prop:simple} and the (analytic) implicit function theorem imply that $x_k$ are analytic in $t\geq0$. This is certainly stronger than Lemma \ref{lem:cont} but with the extra condition that  $G_{>0}$ is connected.
\end{remark}
An immediate consequence of Proposition \ref{prop:simple} is the following.
\begin{corollary}\label{cor:simple}
	Under the assumptions of Proposition \ref{prop:simple}, we have for $t\in[0,\infty)$,
	\begin{align}
		&0<x_1(t)<x_2(t)<\dots<x_{n/2}(t)<\pi/2 \text{ if } n \text{ is even},\\
		&0<x_1(t)<x_2(t)<\dots<x_{(n-1)/2}(t)<\pi/2=x_{(n+1)/2} \text{ if } n \text{ is odd}.
	\end{align}
\end{corollary}
\begin{proof}
	Proposition \ref{prop:simple}, \eqref{eq:thetadef} and \eqref{eq:xandtheta} imply that $Z_t(x)$ has exactly $n$ simple zeros in $(0,\pi)$. From \eqref{eq:evenzeros1} and \eqref{eq:oddzeros1}, we know both $x_j(t)$ and $\pi-x_j(t)$ are zeros in $(0,\pi)$ if $x_j(t)\in(0,\pi/2)$. Therefore, $\pi/2$ cannot be a simple zero of $Z_t(x)$ if $n$ is even, and it is a simple zero if $n$ is odd.
\end{proof}

The following lemma relates $Z_t(x)$ at different $t$.
\begin{lemma}\label{lem:Zdifft}
	For any finite graph $G=(V,E)$ and any pair interactions $\mathbf{J}$ on $G$,
	\begin{align}
		Z_{t-\delta}(x)=&-\sinh(\delta)\sum_{\sigma\in\{-1,+1\}^{V}}\sigma_{u_0}\sigma_{v_0}\exp\left[t\sigma_{u_0}\sigma_{v_0}+\sum_{uv\in E}J_{uv}\sigma_u\sigma_v+ix\sum_{u\in V}\sigma_u\right]\nonumber\\
		&+\cosh(\delta)Z_t(x), \forall t\in\mathbb{C}, \forall \delta\in\mathbb{C}, \forall x\in\mathbb{C}.
	\end{align}
\end{lemma}
\begin{proof}
	From \eqref{eq:Z_tdef}, it is easy to see that $Z_t(x)$ is analytic in $t$ and $x$. So for any fixed $t_0\in\mathbb{C}$ and $x\in\mathbb{C}$, the Taylor expansion of $Z_{t}(x)$ at $t_0$ gives
	\begin{equation}\label{eq:Ztaylor}
		Z_{t_0-\delta}(x)=\sum_{k=0}^{\infty}\left.\frac{\partial^{2k+1}Z_t(x)}{\partial t^{2k+1}}\right|_{t=t_0}\frac{(-\delta)^{2k+1}}{(2k+1)!}+\sum_{k=0}^{\infty}\left.\frac{\partial^{2k}Z_t(x)}{\partial t^{2k}}\right|_{t=t_0}\frac{(-\delta)^{2k}}{(2k)!}.
	\end{equation}
Note that for any $k\in\mathbb{N}\cup\{0\}$,
\begin{align}
	&\frac{\partial^{2k}Z_t(x)}{\partial t^{2k}}=Z_t(x),\\
	&\frac{\partial^{2k+1}Z_t(x)}{\partial t^{2k+1}}=\sum_{\sigma\in\{-1,+1\}^{V}}\sigma_{u_0}\sigma_{v_0}\exp\left[t\sigma_{u_0}\sigma_{v_0}+\sum_{uv\in E}J_{uv}\sigma_u\sigma_v+ix\sum_{u\in V}\sigma_u\right].
\end{align}
The lemma follows by plugging the last two equations into \eqref{eq:Ztaylor}.
\end{proof}

We make the following monotonic change of variables for each $t\in[0,\infty)$
\begin{equation}\label{eq:yandx}
	y_k(t):=\sin^2 x_k(t), \forall k\in\{1,\dots,n/2\} \text{ if }n \text{ even}, \forall k\in\{1,\dots,(n+1)/2\} \text{ if }n \text{ odd}.
\end{equation}

The following lemma gives a relation for $y_k$ at different times.
\begin{lemma}\label{lem:ydifft}
	Let $G=(V,E)$ be a finite graph and  $\mathbf{J}$  be ferromagnetic pair interactions on $G$. If $n$ is even, then for each $k\in\{1,2,\dots,n/2\}$, we have
	\begin{equation}\label{eq:y_kdt}
		\prod_{j=1}^{n/2}\frac{y_j(0)-y_k(t)}{y_j(0)}=-K(t,s)\prod_{j=1}^{n/2}\frac{y_j(s)-y_k(t)}{y_j(s)}, \forall t,\forall s \text{ satisfying } 0<t<s,
	\end{equation}
where 
\begin{equation}
	K(t,s):=\frac{Z_s(0)\sinh(t)}{Z_0(0)\sinh(s-t)}>0.
\end{equation}
If $n$ is odd, the last two displayed equations still hold for each $k\in\{1,2,\dots,(n-1)/2\}$ but one needs to replace $n$ by $n-1$ in \eqref{eq:y_kdt}.
\end{lemma}
\begin{proof}
	By setting $s=t-\delta$ and $x=x_k(t)$ in Lemma \ref{lem:Zdifft}, we have
	\begin{equation}
		Z_s(x_k(t))=-\sinh(t-s)\sum_{\sigma\in\{-1,+1\}^{V}}\sigma_{u_0}\sigma_{v_0}\exp\left[t\sigma_{u_0}\sigma_{v_0}+\sum_{uv\in E}J_{uv}\sigma_u\sigma_v+ix_k(t)\sum_{u\in V}\sigma_u\right].
	\end{equation}
Therefore,
\begin{equation}
	\frac{Z_0(x_k(t))}{\sinh(t)}=\frac{Z_s(x_k(t))}{\sinh(t-s)}, \forall t> 0, \forall s\geq0 \text{ and }s\neq t.
\end{equation}
This, combined with Lemma \ref{lem:fac} and \eqref{eq:yandx}, completes the proof of the lemma.
\end{proof}

We next derive a system of ordinary differential equations for $\{y_k(t)\}$.
\begin{lemma}\label{lem:yODE}
	Let $G=(V,E)$ be a finite graph and $\mathbf{J}$ be ferromagnetic pair interactions on $G$. Suppose that $G_{>0}=(V,E_{>0})$ defined by \eqref{eq:E+} is connected. If $n$ is even, then for each $k\in\{1,2,\dots,n/2\}$, we have
	\begin{equation}\label{eq:yODE}
		y^{\prime}_{k}(t)=-\frac{1}{\sinh(t)}\frac{y_k(t)}{\langle\sigma_{u_0}\sigma_{v_0}\rangle_{\mathbf{J}}\sinh(t)+\cosh(t)}\prod_{j=1}^{n/2}\frac{y_j(0)-y_k(t)}{y_j(0)}\prod_{j=1,j\neq k}^{n/2}\frac{y_j(t)}{y_j(t)-y_k(t)}, \forall t>0, 
	\end{equation}
where 
\begin{equation}
	\langle\sigma_{u_0}\sigma_{v_0}\rangle_{\mathbf{J}}:=\frac{\sum_{\sigma\in\{-1,+1\}^V}\sigma_{u_0}\sigma_{v_0}\exp\left[\sum_{uv\in E}J_{uv}\sigma_u\sigma_v\right]}{\sum_{\sigma\in\{-1,+1\}^V}\exp\left[\sum_{uv\in E}J_{uv}\sigma_u\sigma_v\right]}.
\end{equation}
If $n$ is odd, the last two displayed equations still hold for each $k\in\{1,2,\dots,(n-1)/2\}$ but one needs to replace $n$ by $n-1$ in \eqref{eq:yODE}.
\end{lemma}
\begin{proof}
	We only consider the case that $n$ is even since the proof for the odd case is similar.
	By Proposition \ref{prop:simple}, $x_k(t)$ is a simple zero of $Z_t(x)$ for any $t\geq0$. So 
	\begin{equation}
		Z^{\prime}_t(x_k(t)):=\left.\frac{\partial Z_t(x)}{\partial x}\right|_{x=x_k(t)}\neq 0, \forall t\geq0.
	\end{equation}
The implicit function theorem then implies that $x_k(t)$ is continuously differentiable in $t>0$. So we may differentiate the following equation with respect to $t$:
\begin{equation}
	Z_t(x_k(t))=\sum_{\sigma\in\{-1,+1\}^{V}}\exp\left[t\sigma_{u_0}\sigma_{v_0}+\sum_{uv\in E}J_{uv}\sigma_u\sigma_v+ix_k(t)\sum_{u\in V}\sigma_u\right]\equiv0.
\end{equation}
As a result, we get
\begin{equation}\label{eq:Zdiffx}
	x^{\prime}_k(t)Z_t^{\prime}(x_k(t))=-\sum_{\sigma\in\{-1,+1\}^{V}}\sigma_{u_0}\sigma_{v_0}\exp\left[t\sigma_{u_0}\sigma_{v_0}+\sum_{uv\in E}J_{uv}\sigma_u\sigma_v+ix_k(t)\sum_{u\in V}\sigma_u\right].
\end{equation}
An easy calculation using Lemma \ref{lem:fac} gives
\begin{equation}\label{eq:Zdiffx1}
	Z_t^{\prime}(x_k(t))=-2Z_t(0)\cot(x_k(t))\prod_{j=1,j\neq k}^{n/2}\frac{\sin^2x_j(t)-\sin^2x_k(t)}{\sin^2x_j(t)}.
\end{equation}
Lemma \ref{lem:Zdifft} with $\delta=t$ and $x=x_k(t)$ and Lemma \ref{lem:fac} give
\begin{align}\label{eq:Zdiffx2}
	&-\sum_{\sigma\in\{-1,+1\}^{V}}\sigma_{u_0}\sigma_{v_0}\exp\left[t\sigma_{u_0}\sigma_{v_0}+\sum_{uv\in E}J_{uv}\sigma_u\sigma_v+ix_k(t)\sum_{u\in V}\sigma_u\right]=\frac{Z_0(x_k(t))}{\sinh(t)}\nonumber\\
	&\qquad=\frac{Z_0(0)}{\sinh(t)}\prod_{j=1}^{n/2}\frac{\sin^2x_j(0)-\sin^2x_k(t)}{\sin^2x_j(0)}.
\end{align}
Combining \eqref{eq:Zdiffx}, \eqref{eq:Zdiffx1} and \eqref{eq:Zdiffx2}, we get
\begin{equation}\label{eq:x_kdiff}
	x_k^{\prime}(t)=\frac{Z_0(0)\tan x_k(t)}{-2Z_t(0)\sinh(t)}\prod_{j=1}^{n/2}\frac{\sin^2x_j(0)-\sin^2x_k(t)}{\sin^2x_j(0)}\prod_{j=1,j\neq k}^{n/2}\frac{\sin^2x_j(t)}{\sin^2x_j(t)-\sin^2x_k(t)}, \forall t>0.
\end{equation}
Note that
\begin{align}
	\frac{Z_t(0)}{Z_0(0)}&=\frac{\sum_{\sigma\in\{-1,+1\}^V}\exp\left[t\sigma_{u_0}\sigma_{v_0}+\sum_{uv\in E}J_{uv}\sigma_u\sigma_v\right]}{\sum_{\sigma\in\{-1,+1\}^V}\exp\left[\sum_{uv\in E}J_{uv}\sigma_u\sigma_v\right]}=\left\langle\exp\left[t\sigma_{u_0}\sigma_{v_0}\right]\right\rangle_{\mathbf{J}}\\
	&=\left\langle\sum_{k=0}^{\infty}\frac{\sigma_{u_0}\sigma_{v_0}t^{2k+1}}{(2k+1)!}+\sum_{k=0}^{\infty}\frac{t^{2k}}{(2k)!}\right\rangle_{\mathbf{J}}=\langle\sigma_{u_0}\sigma_{v_0}\rangle_{\mathbf{J}}\sinh(t)+\cosh(t).
\end{align}
This, \eqref{eq:yandx} and \eqref{eq:x_kdiff} complete the proof of the lemma.
\end{proof}

We derive the locations of $y_k(t)$ in the next proposition.
\begin{proposition}\label{prop:y_kloc}
	Let $G=(V,E)$ be a finite graph and $\mathbf{J}$ be ferromagnetic pair interactions on $G$. Suppose that $G_{>0}=(V,E_{>0})$ defined by \eqref{eq:E+} is connected. If $n$ is even, then for each fixed $s\in(0,\infty)$ and $k\in\{1,2,\dots,n/2\}$, exactly one of the following three cases occurs
	\begin{enumerate}[a)]
		\item $y_k(t)\equiv y_k(0)$ for each $t\in[0,s]$,
		\item $y_k(s)\in(y_{k-1}(0),y_k(0))$, $y_k(t)\in(y_k(s),y_k(0))$ for each $t\in(0,s)$,
		\item $y_k(s)\in(y_{k}(0),y_{k+1}(0))$, $y_k(t)\in(y_k(0),y_k(s))$ for each $t\in(0,s)$,
	\end{enumerate}
where we set $y_0(0):=0$ and $y_{n/2+1}(0):=1$. If $n$ is odd, the above statement still holds for each $k\in\{1,2,\dots,(n-1)/2\}$.
\end{proposition}
\begin{proof}
	We only consider the case where $n$ is even since the proof for the odd case is similar.
	Note that Lemma \ref{lem:cont} and \eqref{eq:yandx} imply that $y_k(t)$ is continuous in $t\in[0,\infty)$ for any $k\in\{1,2,\dots,n/2\}$, and Corollary \ref{cor:simple} and \eqref{eq:yandx} imply that
	\begin{equation}\label{eq:ysimple}
		0<y_1(t)<y_2(t)<\dots<y_{n/2}(t)<1, \forall t\in[0,\infty).
	\end{equation}
	We first prove the following claim.
\begin{claim}\label{clm:ykint}
	\begin{equation}
		y_k(s)\in[y_{k-1}(0),y_{k+1}(0)], \forall k\in\{1,2,\dots,n/2\}.
	\end{equation}
\end{claim}
\begin{proof}
We prove this claim by induction. From \eqref{eq:ysimple}, we know $y_1(s)>0=y_0(0)$. Suppose that $y_1(s)>y_2(0)$. Then by the continuity of $y_1(t)$, there exists $T\in(0,s)$ such that $y_1(T)\in(y_2(0),y_1(s))$ and $y_1(T)\in(y_2(0),y_3(0))$. Lemma \ref{lem:ydifft} gives that
\begin{equation}
	\prod_{j=1}^{n/2}\frac{y_j(0)-y_1(T)}{y_j(0)}=-K(T,s)\prod_{j=1}^{n/2}\frac{y_j(s)-y_1(T)}{y_j(s)}, \text{ with }K(T,s)>0.
\end{equation}
It is clear that the LHS of the above displayed equation is positive, while the RHS is negative. This contradiction implies that $y_1(s)\leq y_2(0)$. Now suppose $y_j(s)\in[y_{j-1}(0),y_{j+1}(0)]$ for each $j\in\{1,2,\dots,k\}$ with $k<n/2$. We will prove $y_{k+1}(s)\in[y_k(0),y_{k+2}(0)]$ by contradiction. Suppose that $y_{k+1}(s)<y_k(0)$. By the continuity of $y_{k+1}(t)$, there exists $T\in(0,s)$ such that $y_{k+1}(T)\in(y_{k+1}(s),y_k(0))$ and $y_{k+1}(T)\in (y_{k+1}(s),y_{k+2}(s))$. By the induction hypothesis, $y_k(s)\geq y_{k-1}(0)$, so
\begin{equation}
	y_{k+1}(T)\in(y_{k+1}(s),y_k(0))\subset(y_k(s),y_k(0))\subseteq (y_{k-1}(0),y_k(0)).
\end{equation}
Lemma \ref{lem:ydifft} gives that
\begin{equation}\label{eq:yk+1T}
	\prod_{j=1}^{n/2}\frac{y_j(0)-y_{k+1}(T)}{y_j(0)}=-K(T,s)\prod_{j=1}^{n/2}\frac{y_j(s)-y_{k+1}(T)}{y_j(s)}, \text{ with }K(T,s)>0.
\end{equation}
The LHS of the last displayed equation has the sign $(-1)^{k-1}$, while the the RHS has the sign $(-1)(-1)^{k+1}$. This contradiction implies that $y_{k+1}(s)\geq y_k(0)$. Next suppose that $y_{k+1}(s)>y_{k+2}(0)$. Then we may assume $k\leq n/2-2$ because of \eqref{eq:ysimple}. By the continuity of $y_{k+1}(t)$, there exists $T\in(0,s)$ such that $y_{k+1}(T)\in(y_{k+2}(0),y_{k+1}(s))$ and $y_{k+1}(T)\in(y_{k+2}(0),y_{k+3}(0))$. By the induction hypothesis, $y_k(s)\leq y_{k+1}(0)$, so
\begin{equation}
	y_{k+1}(T)\in(y_{k+2}(0),y_{k+1}(s))\subset (y_{k+1}(0),y_{k+1}(s))\subseteq (y_k(s),y_{k+1}(s)).
\end{equation}
This time the LHS of \eqref{eq:yk+1T} has the sign $(-1)^{k+2}$, while the RHS has the sign $(-1)(-1)^k$. This contradiction implies that $y_{k+1}(s)\leq y_{k+2}(0)$. This completes the proof of the claim. 
\end{proof}
We next prove the following claim.
\begin{claim}\label{clm:ykend}
\begin{equation}
	y_k(s)\notin\{y_{k-1}(0),y_{k+1}(0)\}, \forall k\in\{1,2,\dots,n/2\}.
\end{equation}
\end{claim}
\begin{proof}
Suppose that $y_k(s)=y_{k-1}(0)$. Then Lemma \ref{lem:ydifft} implies that
\begin{equation}\label{eq:yk-1t}
	\prod_{j=1,j\neq k-1}^{n/2}\frac{y_j(0)-y_{k-1}(t)}{y_j(0)}=-K(t,s)\prod_{j=1,j\neq k}^{n/2}\frac{y_j(s)-y_{k-1}(t)}{y_j(s)}, \forall 0<t<s.
\end{equation}
Note that 
\begin{equation}
	\lim_{t\downarrow 0} K(t,s)=\lim_{t\downarrow 0}\frac{Z_s(0)\sinh(t)}{Z_0(0)\sinh(s-t)}=0.
\end{equation}
So by letting $t\downarrow0$ in \eqref{eq:yk-1t}, we get that the RHS of \eqref{eq:yk-1t} is 0, and thus from the LHS we have
\begin{equation}
	y_{k-1}(0)=\lim_{t\downarrow 0}y_{k-1}(t)\in\{y_1(0),y_2(0),\dots,y_{k-2}(0),y_{k}(0),\dots,y_{n/2}(0)\},
\end{equation}
which contradicts \eqref{eq:ysimple}. Suppose that $y_k(s)=y_{k+1}(0)$. Then Lemma \ref{lem:ydifft} implies that
\begin{equation}\label{eq:yk+1t}
	\prod_{j=1,j\neq k+1}^{n/2}\frac{y_j(0)-y_{k+1}(t)}{y_j(0)}=-K(t,s)\prod_{j=1,j\neq k}^{n/2}\frac{y_j(s)-y_{k+1}(t)}{y_j(s)}, \forall 0<t<s.
\end{equation}
By letting $t\downarrow0$ in \eqref{eq:yk+1t}, we get
\begin{equation}
	y_{k+1}(0)=\lim_{t\downarrow 0}y_{k+1}(t)\in\{y_1(0),y_2(0),\dots,y_{k}(0),y_{k+2}(0),\dots,y_{n/2}(0)\},
\end{equation}
which again contradicts \eqref{eq:ysimple}. 
\end{proof}
Combining Claims \ref{clm:ykint} and \ref{clm:ykend}, we have
	\begin{equation}
	y_k(s)\in(y_{k-1}(0),y_{k+1}(0)), \forall k\in\{1,2,\dots,n/2\}.
\end{equation}
We now divide the possible range of $y_k(s)$ into the following three disjoint subsets.
\begin{enumerate} [\text{Case} 1:]
	\item $y_k(s)=y_k(0)$. Suppose that there exists $t\in(0,s)$ such that $y_k(t)>y_k(0)$. Then by the continuity of $y_k(t)$, one can find $T\in(0,s)$ such that $y_k(T)\in(y_k(0),y_{k+1}(0))$ and $y_k(T)\in(y_k(s),y_{k+1}(s))$. Lemma \ref{lem:ydifft} gives that
	\begin{equation}\label{eq:ykT}
		\prod_{j=1}^{n/2}\frac{y_j(0)-y_k(T)}{y_j(0)}=-K(T,s)\prod_{j=1}^{n/2}\frac{y_j(s)-y_k(T)}{y_j(s)}, \text{ with }K(T,s)>0.
	\end{equation}
The LHS of \eqref{eq:ykT} has the sign $(-1)^k$, and the RHS of \eqref{eq:ykT} has the sign $(-1)(-1)^k$, which is a contradiction. Similarly, one can prove that $y_k(t)$ cannot be less than $y_k(0)$. Therefore, $y_k(t)=y_k(0)$ for each $t\in[0,s]$.
	\item $y_k(s)\in(y_{k-1}(0),y_k(0))$. Suppose that there exists $t\in(0,s)$ such that $y_k(t)>y_k(0)$. Then by the continuity of $y_k(t)$, one can find $T\in(0,s)$ such that $y_k(T)\in(y_k(0),y_{k+1}(0))$ and $y_k(T)\in(y_k(s),y_{k+1}(s))$. The same argument as in Case~1 would lead to a contradiction. A similar argument shows that $y_k(t)$ cannot be less than $y_k(s)$ for each $t\in(0,s)$. Therefore, $y_k(t)\in(y_k(s),y_k(0))$ for each $t\in(0,s)$.
	\item $y_k(s)\in(y_k(0),y_{k+1}(0))$. Similar arguments as in Case 2 show that $y_k(t)\in(y_k(0),y_k(s))$ for each $t\in(0,s)$.
\end{enumerate}
This completes the proof of the proposition.
\end{proof}
We are ready to prove Theorem \ref{thm:t}.
\begin{proof}[Proof of Theorem \ref{thm:t}]
	We only consider the case that $n$ is even since the proof for the odd case is similar.
	First, \eqref{eq:xsimple} follows from Corollary \ref{cor:simple}. Proposition \ref{prop:y_kloc} and the continuity of $y_k(s)$ in $s\in[0,\infty)$ imply that exactly one of the following three cases occurs:
	\begin{enumerate}[a)]
		\item $y_k(t)\equiv y_k(0)$ for each $t\in[0,\infty)$,
		\item $y_k(t)\in(y_{k-1}(0),y_k(0))$ for each $t\in(0,\infty)$,
		\item $y_k(t)\in(y_{k}(0),y_{k+1}(0))$ for each $t\in(0,\infty)$.
	\end{enumerate}
By Lemma \ref{lem:yODE}, we have
\begin{equation}\label{eq:ykprime}
	y^{\prime}_{k}(t)=-\frac{1}{\sinh(t)}\frac{y_k(t)}{\langle\sigma_{u_0}\sigma_{v_0}\rangle_{\mathbf{J}}\sinh(t)+\cosh(t)}\prod_{j=1}^{n/2}\frac{y_j(0)-y_k(t)}{y_j(0)}\prod_{j=1,j\neq k}^{n/2}\frac{y_j(t)}{y_j(t)-y_k(t)}, \forall t>0.
\end{equation}
Obviously, we have
\begin{equation}
	|\langle\sigma_{u_0}\sigma_{v_0}\rangle_{\mathbf{J}}|\leq 1.
\end{equation}
Note that Corollary \ref{cor:simple} and \eqref{eq:yandx} imply that
\begin{equation}
	0<y_1(t)<y_2(t)<\dots<y_{n/2}(t)<1, \forall t\in[0,\infty).
\end{equation}
Thus, we have
\begin{equation}
	\frac{1}{\sinh(t)}\frac{y_k(t)}{\langle\sigma_{u_0}\sigma_{v_0}\rangle_{\mathbf{J}}\sinh(t)+\cosh(t)}>0, \forall t>0.
\end{equation}
If $y_k(t)\in(y_{k-1}(0),y_k(0))$ for each $t\in(0,\infty)$, then
\begin{equation}
	\text{the sign of }\prod_{j=1}^{n/2}\frac{y_j(0)-y_k(t)}{y_j(0)}=(-1)^{k-1}, \text{the sign of }\prod_{j=1,j\neq k}^{n/2}\frac{y_j(t)}{y_j(t)-y_k(t)}=(-1)^{k-1}.
\end{equation}
Therefore, \eqref{eq:ykprime} implies that $y_k^{\prime}(t)<0$ for each $t\in(0,\infty)$. 

Similarly, if $y_k(t)\in(y_{k}(0),y_{k+1}(0))$ for each $t\in(0,\infty)$, then one can prove  that $y_k^{\prime}(t)>0$ for each $t\in(0,\infty)$. These combined with \eqref{eq:yandx} complete the proof of the theorem.
\end{proof}

\section{Monotonicity on the complete graph}\label{sec:K_n}
In this section, we prove Theorem \ref{thm:K_n}. Since we always fix $n\in\mathbb{N}$ throughout this section, to simplify the notation, we may drop $n$ from our notation and write
\begin{equation}\label{eq:Z_tdef1}
	Z_t(x):=Z_{n,t}(x)=\sum_{\sigma\in\{-1,+1\}^{n}}\exp\left[t\left(\sum_{j=1}^n\sigma_j\right)^2+ix\sum_{j=1}^n\sigma_j\right].
\end{equation}
For fixed $t\geq 0$, it is well-known that $Z_t(x)$ is in the Laguerre-P\'{o}lya class (see, e.g., \cite{CSV94} for the definition). Actually, we have
\begin{lemma}\label{lem:LP}
	For any $t\geq 0$,  $Z_t(x)$ is in the Laguerre-P\'{o}lya class and
	\begin{equation}
	Z_t(x)=Z_t(0)\prod_{j=1}^{\infty}\left(1-\frac{x^2}{x_j^2(t)}\right)=
	\begin{cases}
		Z_t(0)\prod_{j=1}^{n/2}\frac{\sin^2x_j(t)-\sin^2x}{\sin^2x_j(t)}, &\text{if } n \text{ even}\\
		Z_t(0)\cos(x)\prod_{j=1}^{(n-1)/2}\frac{\sin^2x_j(t)-\sin^2x}{\sin^2x_j(t)}, &\text{if } n \text{ odd},
	\end{cases}
\end{equation}
where $0<x_1(t)\leq x_2(t)\leq\dots$ are all positive zeros of $Z_t(x)$ (listed according to their multiplicities), and $\sum_{j=1}^{\infty}x_j^{-2}(t)<\infty$.
\end{lemma}
\begin{proof}
	The proof is the same as that of Lemma \ref{lem:fac}.
\end{proof}

It is easy to check that each zero of $Z_0(x)$ has multiplicity $n$, but we will see that all zeros of $Z_t(x)$ are simple once $t>0$.
\begin{lemma}\label{lem:Zsimple}
	If $t>0$, then all zeros of $Z_t(x)$ (as a function of $x$ with $t$ fixed) are simple.
\end{lemma}
\begin{proof}
	Of course, the lemma follows from Proposition \ref{prop:simple}. We give a different proof which utilizes the fact that $Z_t(x)$ satisfies the backward heat equation \eqref{eq:bhe} below. We follow the idea of the proof of Lemma 2.2 in \cite{CSV94}: if $x_0$ were a zero of $Z_{t_0}$ of multiplicity at least $2$, then $Z_{t_0-\delta}$ would not be in the Laguerre-P\'{o}lya class for all small $\delta$. It is clear that $Z_t(x)$ is analytic in $(t,x)$. It is also easy to check from \eqref{eq:Z_tdef1} that $Z_t(x)$ satisfies the backward heat equation:
	\begin{equation}\label{eq:bhe}
		\frac{\partial Z_t(x)}{\partial t}=-\frac{\partial^2 Z_t(x)}{\partial x^2}.
	\end{equation}
	From this, we get that for any $t_0>0$ and $\delta>0$,
	\begin{equation}
		Z_{t_0-\delta}(x)=\sum_{j=0}^{\infty}\left.\frac{\partial^j Z_t(x)}{\partial t^j}\right|_{t=t_0}\frac{(-\delta)^j}{j!}=\sum_{j=0}^{\infty} Z_{t_0}^{(2j)}(x)\frac{\delta^j}{j!}, \forall x\in\mathbb{C},
	\end{equation}
where $()$ in the superscript of $Z_{t_0}$ denotes the partial derivative with respect to $x$. Therefore, for any $l\in\mathbb{N}\cup\{0\}$,
	\begin{equation}\label{eq:ZTaylor}
		Z^{(l)}_{t_0-\delta}(x)=\sum_{j=0}^{\infty} Z_{t_0}^{(2j+l)}(x)\frac{\delta^j}{j!},\forall x\in\mathbb{C}.
	\end{equation}
	Suppose that for some $t_0>0$, $x_0\in\mathbb{R}$ is a zero of $Z_{t_0}$ of multiplicity $k+1$ with $k\in\mathbb{N}$, i.e.,
	\begin{equation}\label{eq:mulzeros}
		Z_{t_0}^{(j)}(x_0)=0, \forall j\in\{0,1,\dots,k\}, 	Z_{t_0}^{(k+1)}(x_0)\neq0.
	\end{equation}
	We consider the function 
	\begin{equation}
		z_t(x):=Z_t^{(k-1)}(x).
	\end{equation}
	Then \eqref{eq:mulzeros} implies that
	\begin{equation}
		z_{t_0}(x_0)=z^{\prime}_{t_0}(x_0)=0, z^{\prime\prime}_{t_0}(x_0)\neq0.
	\end{equation}
	Combining this with \eqref{eq:ZTaylor}, we have
	\begin{equation}
		L_1\left(z_{t_0-\delta}(x_0)\right):=\left[z^{\prime}_{t_0-\delta}(x_0)\right]^2-z_{t_0-\delta}(x_0)z^{\prime\prime}_{t_0-\delta}(x_0)=-\delta\left[z^{\prime\prime}_{t_0}(x_0)\right]^2+O(\delta^2) \text{ as }\delta\downarrow 0.
	\end{equation}
	Therefore,
	\begin{equation}\label{eq:zL_1}
		L_1\left(z_{t_0-\delta}(x_0)\right)<0 \text{ for all small enough }\delta>0.
	\end{equation}
	It is known (see, e.g., the remarks after Theorem 2.9 of \cite{CV90}) that for any $f$ in the Laguerre-P\'{o}lya class, $f$ satisfies
	\begin{equation}
		L_m(f(x)):=\left[f^{(m)}(x)\right]^2-f^{(m-1)}(x)f^{(m+1)}(x)\geq 0, \forall m\in\mathbb{N}, \forall x\in\mathbb{R}.
	\end{equation}
	Thus, \eqref{eq:zL_1} implies that $z_{t_0-\delta}=Z_{t_0-\delta}^{(k-1)}$ is not in the Laguerre-P\'{o}lya class for all small $\delta>0$. However, the Laguerre-P\'{o}lya class is closed under differentiation  by the Gauss-Lucas theorem (see, e.g., Theorem 6.2 on p. 22 of \cite{Mar66}). This contradiction implies that our hypothesis \eqref{eq:mulzeros} is false and thus completes the proof of the lemma.
\end{proof}

The following corollary is analogous to Corollary \ref{cor:simple}.
\begin{corollary}\label{cor:xinD}
	If $n$ is even,
	\begin{equation}
		0<x_1(t)<x_2(t)<\dots<x_{n/2}(t)<\pi/2, \forall t>0;
	\end{equation}
	if $n$ is odd,
	\begin{equation}
		0<x_1(t)<x_2(t)<\dots<x_{(n-1)/2}(t)<\pi/2=x_{(n+1)/2}(t), \forall t>0.
	\end{equation}
\end{corollary}
\begin{proof}
	The proof is the same as that of Corollary \ref{cor:simple}.
\end{proof}

We are in position to study the dynamics of the zeros of $Z_t(x)$.  
\begin{lemma}\label{lem:ode}
	If $n$ is even, the first $n/2$ positive zeros of $Z_t(x)$ defined in Lemma \ref{lem:LP} satisfy that for each $ k\in\{1,2,\dots,n/2\}$,
	\begin{equation}\label{eq:odeeven}
		x_k^{\prime}(t)=2\cot\left(2x_k(t)\right)-2\sin(2x_k(t))\sum_{j=1,j\neq k}^{n/2}\left[\sin^2x_j(t)-\sin^2x_k(t)\right]^{-1}, \forall t>0.
	\end{equation}
	If $n$ is odd, the first $(n+1)/2$ positive zeros of $Z_t(x)$ defined in Lemma \ref{lem:LP} satisfy that for each $ k\in\{1,2,\dots,(n-1)/2\}$,
	\begin{align}
		&x_k^{\prime}(t)=2\cot\left(2x_k(t)\right)-2\sin(2x_k(t))\sum_{j=1,j\neq k}^{(n+1)/2}\left[\sin^2x_j(t)-\sin^2x_k(t)\right]^{-1}, \forall t>0,\\
		&x_{(n+1)/2}(t)\equiv\pi/2, \forall t\geq 0.
	\end{align}
\end{lemma}
\begin{proof}
    We only consider the case that $n$ is even since the proof for the odd case is similar.
	Lemma \ref{lem:Zsimple} says that $x_k(t)$ is a simple zero of $Z_t$ for any $t>0$. The (analytic) implicit function theorem then implies that $x_k(t)$ is analytic in $t>0$. So if we differentiate the following equation with respect to $t$, 
	\begin{equation}
		Z_t(x_k(t))=\sum_{\sigma\in\{-1,+1\}^{n}}\exp\left[t\left(\sum_{j=1}^n\sigma_j\right)^2+ix_k(t)\sum_{j=1}^n\sigma_j\right]=0,
	\end{equation}
	we get
	\begin{equation}\label{eq:Zdiffx3}
		Z_t^{\prime\prime}(x_k(t))=x_k^{\prime}(t)Z_t^{\prime}(x_k(t)), \forall t>0.
	\end{equation}
By Lemma \ref{lem:LP}, we have
	 \begin{align}
	 	&Z_t^{\prime}(x_k(t))=-2Z_t(0)\cot(x_k(t))\prod_{j=1,j\neq k}^{n/2}\frac{\sin^2 x_j(t)-\sin^2 x_k(t)}{\sin^2 x_j(t)},\\
	 	&Z_t^{\prime\prime}(x_k(t))=Z_t(0)\Bigg\{\frac{-2\cos(2x_k(t))}{\sin^2 x_k(t)}\prod_{j=1,j\neq k}^{n/2}\frac{\sin^2 x_j(t)-\sin^2 x_k(t)}{\sin^2 x_j(t)}+\nonumber\\
	 	&\qquad2\frac{\sin^2(2x_k(t))}{\sin^2 x_k(t)}\sum_{j=1,j\neq k}^{n/2}\frac{1}{\sin^2 x_j(t)}\prod_{l=1,l\neq j,k}^{n/2}\frac{\sin^2 x_l(t)-\sin^2 x_k(t)}{\sin^2 x_l(t)}\Bigg\}.
	 \end{align}
	 These combined with \eqref{eq:Zdiffx3} complete the proof of \eqref{eq:odeeven}.
\end{proof}

Assume that $n$ is even. Let $D\subseteq\mathbb{R}^{n/2}$ be the open set
\begin{equation}\label{eq:Ddef}
	D:=\{\mathbf{y}=(y_1,y_2,\dots,y_{n/2}):0<y_1<y_2<\dots<y_{n/2}<\pi/2\}.
\end{equation}
By Lemma \ref{lem:ode}, we know $\mathbf{x}:=(x_1(t),x_2(t),\dots,x_{n/2}(t))$ satisfies the autonomous system of ODEs
\begin{equation}\label{eq:odef}
	\mathbf{x}^{\prime}=f(\mathbf{x}),
\end{equation}
where $f:D\rightarrow\mathbb{R}^{n/2}$, $\mathbf{x}\mapsto(f_1(\mathbf{x}),f_2(\mathbf{x}),\dots,f_{n/2}(\mathbf{x}))$, is defined by
\begin{equation}\label{eq:f_kdef}
	f_k(\mathbf{x}):=2\cot(2x_k)-2\sin(2x_k)\sum_{j=1,j\neq k}^{n/2}\left[\sin^2x_j-\sin^2x_k\right]^{-1}, k\in\{1,2,\dots,n/2\}.
\end{equation}
It is clear that $f$ is continuously differentiable on $D$. A key observation is
\begin{equation}\label{eq:f_kdergeq0}
	\frac{\partial f_k}{\partial x_j}(\mathbf{x})=2\sin(2x_k)\sin(2x_j)\left[\sin^2 x_j-\sin^2 x_k\right]^{-2}\geq 0, \forall j\neq k, \forall \mathbf{x}\in D.
\end{equation}
Then the dynamical system \eqref{eq:odef} is a so-called \textit{cooperative system} (see Chapter 3 of \cite{Smi95} for more information).
We use the partial order on $\mathbb{R}^{n/2}$ defined by $\mathbf{x}\leq \mathbf{y}$ if $\mathbf{y}-\mathbf{x}\in\mathbb{R}_{\geq 0}^{n/2}$ where $\mathbb{R}_{\geq 0}$ is the set of nonnegative real numbers. We write $\mathbf{x}\ll\mathbf{y}$ if $x_k<y_k$ for each $1\leq k\leq n/2$. By Remark 1.1 on p. 33 of \cite{Smi95}, \eqref{eq:Ddef} and \eqref{eq:f_kdergeq0} imply that $f$ is of \textit{type K} in $D$; i.e., $f_k(\mathbf{x})\leq f_k(\mathbf{y})$ for any $\mathbf{x},\mathbf{y}\in D$ satisfying $\mathbf{x}\leq \mathbf{y}$ and $x_k=y_k$. The following property about monotone dynamical systems from \cite{Smi95} will be very important to our proof of Theorem \ref{thm:K_n}.
\begin{proposition}[Proposition 1.1 on p. 32 of \cite{Smi95}]\label{prop:Smi}
	Let $g$ be of type K on $D$. Let $\mathbf{y}_0,\mathbf{z}_0$ be in $D$ satisfying $\mathbf{y}_0\ll\mathbf{z}_0$, and $\phi_t(\mathbf{y}_0)$ (respectively, $\phi_t(\mathbf{z}_0)$) be the solution of
	\begin{equation}
		\mathbf{y}^{\prime}=g(\mathbf{y})
	\end{equation}
	starting at $\mathbf{y}_0$ (respectively, $\mathbf{z}_0$) at $t=0$. If $\phi_t(\mathbf{y}_0)$ and $\phi_t(\mathbf{z}_0)$ are defined for some $t>0$, then $\phi_t(\mathbf{y}_0)\ll \phi_t(\mathbf{z}_0)$.
\end{proposition}
We use Proposition \ref{prop:Smi} to prove the following monotone results about $x_k$.
\begin{proposition}\label{prop:mon}
	If $n$ is even, then
	\begin{equation}\label{eq:moneven}
		x_k(t_1)>x_k(t_2), \forall 0\leq t_1<t_2, \forall k\in\{1,2,\dots,n/2\};
	\end{equation}
	If $n$ is odd, then
	\begin{equation}
		x_k(t_1)>x_k(t_2), \forall 0\leq t_1<t_2, \forall k\in\{1,2,\dots,(n-1)/2\}.
	\end{equation}
\end{proposition}
\begin{proof}
	We assume that $n$ is even since the proof for the odd case is similar. Following the notation of Proposition \ref{prop:Smi}, let $\phi_t(\mathbf{y})$ with $\mathbf{y}\in D$ be the solution of \eqref{eq:odef} that starts at the point $\mathbf{y}$ at $t=0$. The existence and uniqueness of $\phi_t(\mathbf{y})$ for any fixed $\mathbf{y}\in D$ and small $|t|$ follows from the Picard-Lindel\"{o}f theorem (see, e.g., Theorem 3.1 on p.~8 of \cite{Hal80}); if $\mathbf{y}=\mathbf{x}(t_0)$ for some $t_0>0$ then Lemma \ref{lem:ode} and \eqref{eq:f_kdef} imply that $\phi_t(\mathbf{y})$ exists (and is unique) for all $t\geq0$.  Let $\mathbf{x}_0:=(x_1(0),x_2(0),\dots,x_{n/2}(0))=(\pi/2,\pi/2,\dots,\pi/2)$. Even though $\mathbf{x}_0\notin D$, we may still define $\phi_t(\mathbf{x}_0)$ by
	\begin{equation}
		\phi_t(\mathbf{x}_0):=\begin{cases}\mathbf{x}_0, &t=0\\
			\phi_{\epsilon}\left(\mathbf{x}(t-\epsilon)\right), &t>0,
		\end{cases}
	\end{equation}
where $\mathbf{x}(t-\epsilon):=(x_1(t-\epsilon),\dots,x_{n/2}(t-\epsilon))$ and $\epsilon\in(0,t)$. Then Lemma~\ref{lem:ode}, Corollary \ref{cor:xinD} and the continuity of $x_k(t)$ in $t\geq 0$ imply that $\phi_t(\mathbf{x}_0)$ is well-defined and is independent of $\epsilon$; moreover
	\begin{equation}
		x_k(t)=\phi_t(\mathbf{x}_0)_k, \forall t\geq 0, \forall k\in\{1,2,\dots,n/2\},
	\end{equation}
	where the subscript $k$ in $\phi_t(\mathbf{x}_0)_k$ denotes the $k$-th component. By Corollary \ref{cor:xinD}, we know
	\begin{equation}\label{eq:phi_tinD}
		\phi_t(\mathbf{x}_0)\in D \text{ for each }t>0. 
	\end{equation} 
	So \eqref{eq:moneven} follows if we can prove
	\begin{equation}\label{eq:mon1}
		\phi_t(\mathbf{x}_0)\gg \phi_t\left(\phi_{\tau}(\mathbf{x}_0)\right)=\phi_{t+\tau}(\mathbf{x}_0), \forall t\geq 0, \forall \tau>0.
	\end{equation}
	The last inequality is trivial if $t=0$. So we only need to prove that for any fixed $t_0>0$ and $\tau_0>0$, we have
	\begin{equation}\label{eq:mon2}
		\phi_{t_0}(\mathbf{x}_0)\gg \phi_{t_0+\tau_0}(\mathbf{x}_0).
	\end{equation}
Indeed, $\phi_0(\mathbf{x}_0)=\mathbf{x}_0\gg  \phi_{\tau_0}(\mathbf{x}_0)$ (see \eqref{eq:phi_tinD}) and the continuity of $	\phi_t(\mathbf{x}_0)$ in $t\geq 0$ imply that there exists $\delta\in(0,t_0]$ such that
	\begin{equation}
		\phi_{\delta}(\mathbf{x}_0)\gg \phi_{\tau_0+\delta}(\mathbf{x}_0).
	\end{equation}
	Applying Proposition \ref{prop:Smi}, we get
	\begin{equation}
		\phi_t\left(\phi_{\delta}(\mathbf{x}_0)\right)\gg \phi_t\left(\phi_{\tau_0+\delta}(\mathbf{x}_0)\right), \forall t\geq0.
	\end{equation}
	Taking $t=t_0-\delta$ in the above displayed inequality, we obtain \eqref{eq:mon2}. This completes the proof of \eqref{eq:mon1} and thus \eqref{eq:moneven}.
\end{proof}

We are ready to prove Theorem \ref{thm:K_n}.
\begin{proof}[Proof of Theorem \ref{thm:K_n}]
	We again only prove the even case since the proof for the odd case is similar. By the definition of $Z_t(x)$ in \eqref{eq:Z_tdef1}, we have
	\begin{equation}
		\lim_{t\rightarrow\infty}\frac{Z_t(x)}{\exp(tn^2)}=\exp(ixn)+\exp(-ixn), x\in\mathbb{C}.
	\end{equation}
	Note that $\{Z_t(x)/\exp(tn^2): t\geq 0\}$ is locally uniformly bounded on $\mathbb{C}$, so by first applying Vitali's convergence theorem and then Hurwitz's theorem, we have
	\begin{equation}\label{eq:xatinfty}
		\lim_{t\rightarrow\infty}x_k(t)=\frac{(2k-1)\pi}{2n}, k\in\{1,2,\dots,n/2\}.
	\end{equation}
	The theorem (for the even case) follows from Proposition \ref{prop:mon} and \eqref{eq:xatinfty}.
\end{proof}

\section*{Acknowledgments}
The research of the first and second authors was partially supported by NSFC grant 11901394.

\bibliographystyle{abbrv}
\bibliography{reference}

\end{document}